\documentclass[conference]{IEEEtran}
\IEEEoverridecommandlockouts
\usepackage{cite}
\usepackage{amsmath,amssymb,amsfonts,texlogos}
\usepackage{algorithmic}
\usepackage{graphicx}
\usepackage{textcomp}
\usepackage{multicol,lipsum}
\usepackage{mathtools, cuted}

\usepackage{xcolor}
\def\BibTeX{{\rm B\kern-.05em{\sc i\kern-.025em b}\kern-.08em
    T\kern-.1667em\lower.7ex\hbox{E}\kern-.125emX}}

 \usepackage[linesnumbered,ruled]{algorithm2e}
\usepackage{amssymb,amsthm,textcomp}

\usepackage{url}
\usepackage[final]{pdfpages}
\usepackage{listings}
\usepackage{textcomp}
\usepackage{color} 
\usepackage{comment}

\usepackage{amsthm}
\usepackage{graphicx} 
\usepackage{mathtools}
\DeclareMathOperator*{\argmax}{arg\,max}
\DeclareMathOperator*{\argmin}{arg\,min}
\usepackage{MnSymbol} 
\usepackage{multirow}
\usepackage{color}
\usepackage{subfigure}

 \usepackage[protrusion=true,expansion=true]{microtype}
\usepackage[section]{placeins}
\usepackage{lipsum}
\usepackage{paralist}
\usepackage{wrapfig}
\theoremstyle{remark}
\usepackage{tabularx}
\newtheorem{definition}{Definition}[]
\newtheorem{theorem}{Theorem}[]

\newtheorem{lemma}{Lemma}[]

\usepackage[font=footnotesize]{caption}

\DeclarePairedDelimiter\floor{\lfloor}{\rfloor}
\makeatletter
\newcommand\footnoteref[1]{\protected@xdef\@thefnmark{\ref{#1}}\@footnotemark}
\makeatother
\renewcommand{\figurename}{Fig.{}}
\renewcommand\IEEEkeywordsname{Keywords}
\begin{document}
\setlength{\abovedisplayskip}{.01pt}
\setlength{\belowdisplayskip}{.01pt}
\title{{Interference Avoidance Position Planning in UAV-assisted Wireless Communication}\thanks{This work was supported in part by the National Science Foundation under grants ECCS-1444009 and CNS-1824518.}}

\author{\IEEEauthorblockN{Seyyedali Hosseinalipour}
\IEEEauthorblockA{\textit{ECE Department} \\
\textit{North Carolina State University}\\
Raleigh, USA \\
shossei3@ncsu.edu}
\and
\IEEEauthorblockN{Ali Rahmati}
\IEEEauthorblockA{\textit{ECE Department} \\
\textit{North Carolina State University}\\
Raleigh, USA \\
arahmat@ncsu.edu}
\and
\IEEEauthorblockN{Huaiyu Dai}
\IEEEauthorblockA{\textit{ECE Department} \\
\textit{North Carolina State University}\\
Raleigh, USA \\
hdai@ncsu.edu}}
\maketitle
\begin{abstract}
We consider unmanned aerial vehicle (UAV)-assisted wireless communication employing UAVs as relay nodes to increase the throughput between a pair of transmitter and receiver. We focus on developing effective methods to position the UAV(s) in the sky in the presence of a major source of interference, the existence of which makes the problem non-trivial. First, we consider utilizing a single UAV, for which we develop a theoretical framework to determine its optimal position aiming to maximize the SIR of the system. To this end, we investigate the problem for three practical scenarios, in which the position of the UAV is: (i) vertically fixed, horizontally adjustable; (ii) horizontally fixed, vertically  adjustable; (iii) both horizontally and vertically adjustable. Afterward, we consider employing multiple UAVs, for which we propose a cost-effective method that simultaneously minimizes the number of required UAVs and determines their optimal positions so as to guarantee a certain SIR of the system. We further develop a distributed placement algorithm, which can increase the SIR of the system given an arbitrary number of UAVs. Numerical simulations are provided to evaluate the performance of our proposed methods. 
\end{abstract}

\section{Introduction}\label{sec:intro}

\noindent Recently, unmanned aerial vehicles (UAVs) have  been considered  as
a promising solution  for a variety of critical applications
such as environmental surveillance, public safety, disaster
relief, search and rescue, and purchase delivery~\cite{hayat2016survey}. 
Constructing a UAV communication network for
such  applications is a non-trivial task since there is no
regulatory and pre-allocated spectrum
band for the UAVs. As a result, this network usually coexists with other communication networks, e.g., cellular networks~\cite{saleem2015integration}.
Thus, studying the problem of interference avoidance for the UAV communication network is critical, where the inherent mobility feature of the UAVs can be deployed as an interference evasion mechanism.

In the UAV literature, there are several works dedicated to placement planning for UAVs (e.g.,~\cite{ladosz2016optimal,chen2018optimum,jiang2018power,zhang2018trajectory,li2018placement,chen2018multiple}).  In~\cite{ladosz2016optimal,chen2018optimum,jiang2018power}, UAVs are utilized in the two-hop relay communication framework, where each UAV directly connects the transmitter/source to the receiver/sink. Each of these works aims to identify the optimal position of UAVs to improve a desired metric of interest, namely the network
connectivity, the reliability, and the throughput of the system. Studying the UAV placement planning in the multi-hop relay communication context, in which multiple UAVs can be utilized between the transmitter and the receiver, is a topic  studied recently in~\cite{zhang2018trajectory,li2018placement,chen2018multiple}. The aim of these works is similar to~\cite{ladosz2016optimal,chen2018optimum,jiang2018power}; however, the existence of communication links among the UAVs makes their methodology different. Nevertheless, none of the aforementioned works consider the placement of UAV(s) in the presence of a source of interference, which is the main motivation for our work. Also, there are some similar works in the literature of sensor networks, among which the most relevant ones are~\cite{roh2010optimal, {chattopadhyay2017deploy}}. In \cite{roh2010optimal}, the two-dimensional (2-D) placement of relays is investigated aiming to increase the achievable transmission rate. In \cite{chattopadhyay2017deploy}, the \textit{impromptu} (as-you-go) placement of the relay nodes between a pair of source and sink node is addressed considering that the distance between those nodes may be a random variable, where the space is restricted to be one-dimensional (1-D).  Nonetheless, none of these works consider a source of interference in the network. 

 Our goal is to go one step further and investigate the UAV-assisted wireless communication paradigm in the presence of a major source of interference (MSI), which refers to the source of interference with the dominant effect in the environment. Our work contributes to the literature by addressing the relay position planning problem in the presence of an MSI in 3-D space. Considering different interpretations for the MSI, e.g., a primary transmitter in UAV cognitive radio networks \cite{saleem2015integration,rahmati2019dynamic}, an eNodeB in UAV-assisted LTE-U/WiFi public safety networks \cite{athukoralage2016regret}, a malicious user in drone delivery application, or a base station in surveillance application, our paper can be adapted to multiple real-world scenarios. We address the UAV placement planning in the presence of an MSI for both two-hop and multi-hop relay communication settings. Our contributions can be summarized as follows: (i) Considering a single UAV, we develop a theoretical approach to identify the optimal position of the UAV so as to maximize the signal-to-interference ratio (SIR) of the system. For this case, we address the design under three circumstances, in which the position of the UAV is either horizontally or vertically fixed, or neither. (ii) For a predetermined SIR of the system, we develop a theoretical framework, which simultaneously determines the minimum required number of UAVs and their optimal positions. (iii) For a predetermined number of UAVs, we propose a distributed algorithm to maximize the SIR of the system, which only requires message exchange among the adjacent UAVs.
\section{Preliminaries}
\vspace{-1mm}
\noindent We consider data transmission between a pair of transmitter (Tx) and receiver (Rx) co-existing with a major source of interference (MSI). We consider a \textit{left-handed coordination system} $(x,y,h)$, where the Tx, the Rx, and the MSI are assumed to be on the ground plane defined as $h=0$. The location of the Tx, the Rx, and the MSI is assumed to be $(0,0,0)$, $(D,0,0)$, and $(X_{_{\textrm{MSI}}},Y_{_{\textrm{MSI}}},0)$, respectively. We assume $0\leq X_{_{\textrm{MSI}}}\leq D$ for simplicity, which can be readily generalized with minor modification. The transmission powers of the Tx, the UAV, and the MSI are denoted by $p_t$, $p_u$, and $p_{_{\textrm{MSI}}}$, respectively. To improve the transmission data rate, it is desired to place a UAV or a set of UAVs, each of which acting as a relay, between the Tx and the Rx. To have tractable derivations, we assume that the UAVs can be placed at $y=0$ plane. Also, considering legal regulations, we confine the altitude of the UAVs to $h\in[h_{min},h_{max}]$. 
		 
 We consider the line-of-sight (LoS) and the non-line-of-sight (NLoS) channel models, for which the path-loss is given by:
\begin{equation}
    L^{\textrm{LoS}}_{i,j}= \mu_{_{\textrm{LoS}}} d_{i,j}^{\alpha},\;\;L^{\textrm{NLoS}}_{i,j}=\mu_{_{\textrm{NLoS}}} d_{i,j}^{\alpha},
\end{equation}
where $\mu_{_{\textrm{LoS}}}\triangleq C_{_{\textrm{LoS}}}\left(4\pi f_c/c\right)^\alpha$, $\mu_{_{\textrm{NLoS}}}\triangleq C_{_{\textrm{NLoS}}}\left(4\pi f_c/c\right)^\alpha$, $C_{_{\textrm{LoS}}}$ ($C_{_{\textrm{NLoS}}}$) is the excessive path loss factor incurred by shadowing, scattering, etc., in LoS (NLoS) link, $f_c$ is the carrier frequency, $c$ is the speed of light, $\alpha=2$ is the path-loss exponent, and $d_{i,j}$ is the Euclidean distance between node $i$ and node $j$. The link between two UAVs (air-to-air) is modeled using the LoS model, while the link between the MSI and the Rx (ground-to-ground) is modeled based on the NLoS model. To model the link between a UAV and the Rx/Tx/MSI (air-to-ground and ground-to-air) either the LoS or the NLoS model~\cite{chen2018multiple,zhang2018trajectory} or a weighted average between the LoS model and the NLoS model~\cite{moza:internet,channel2,channel3} can be used. In this paper, we consider a general case and denote the path loss between a UAV $i$ and node $j$ located on the ground by $\eta_{_{\textrm{NLoS}}}d_{ij}^{2}$. We assume that $\eta_{_{\textrm{NLoS}}}$ is constant in the range $h\in[h_{min},h_{max}]$, and thus $\eta_{_{\textrm{NLoS}}}\triangleq g(\mu_{_{\textrm{LoS}}},\mu_{_{\textrm{NLoS}}},h_{min},h_{max})$, where $g$ is a function. Further discussions on obtaining the $g$ in different environments can be found in~\cite{moza:internet,channel2,channel3}. Due to the geographical limitations, direct communication between the Tx and Rx is not considered. 
\vspace{-1mm}
\section{Position Planning for a Single UAV}\label{sec:singleUAV}
\vspace{-0.5mm}
\noindent Let $\textrm{SIR}_1(x,h)$, $\textrm{SIR}_2(x,h)$ denote the SIR at the UAV located at $(x,0,h)$ and the SIR at the Rx, respectively (see Fig.~\ref{fig:single}), which are given by:
\begin{equation*}
\hspace{-6mm}
\resizebox{0.9\hsize}{!}{$
      \textrm{SIR}_1(x,h)\hspace{-1mm}=\hspace{-1mm}\frac{p_t\hspace{-.14mm}/\hspace{-.15mm}(\hspace{-.15mm}\eta_{_{\textrm{NLoS}}} d_{_{\textrm{UAV},\textrm{Tx}}}^{2}\hspace{-.15mm})}{p_{_{\textrm{MSI}}}\hspace{-.15mm}/\hspace{-.15mm}(\hspace{-.15mm}\eta_{_{\textrm{NLoS}}} d_{_{\textrm{UAV},\textrm{MSI}}}^{2}\hspace{-.15mm})}\hspace{-1mm}=\hspace{-1mm}\frac{p_t \hspace{-.12mm} \left(\hspace{-.15mm}\left(\hspace{-.15mm}x-X_{_{\textrm{MSI}}}\hspace{-.15mm}\right)^2\hspace{-.15mm}+\hspace{-.14mm}Y_{_{\textrm{MSI}}}^2\hspace{-.15mm}+\hspace{-.12mm}h^2\hspace{-.15mm}\right)}{p_{_{\textrm{MSI}}}\left(x^2+h^2\right)},$
      }
\end{equation*}
\begin{equation}
\hspace{1.0mm}
\resizebox{0.9\hsize}{!}{$
      \textrm{SIR}_2\hspace{-.1mm}(\hspace{-.15mm}x,h\hspace{-.15mm})\hspace{-1mm}= \hspace{-.199mm}\frac{\hspace{-.19mm}p_u\hspace{-.15mm}/\hspace{-.15mm}(\hspace{-.15mm}\eta_{_{\textrm{NLoS}}} \hspace{-.15mm}d_{_{\textrm{UAV},\textrm{Rx}}}^{2}\hspace{-.15mm})}{\hspace{-0.17mm}p_{_{\textrm{MSI}}}\hspace{-.15mm}/\hspace{-.15mm}(\hspace{-.15mm}\mu_{_{\textrm{NLoS}}} \hspace{-.15mm}d_{_{\textrm{Rx},\textrm{MSI}}}^{2}\hspace{-.15mm})}\hspace{-1mm}=
      \hspace{-1mm} \frac{\hspace{-.15mm}p_u\hspace{-.1mm} \left(\hspace{-.15mm}Y^2_{_{\textrm{MSI}}}\hspace{-.15mm}+\hspace{-.0mm}(\hspace{-.13mm}D-X_{_{\textrm{MSI}}}\hspace{-.13mm})^2\hspace{-.13mm}\right)}{\hspace{-0.1mm}p_{_{\textrm{MSI}}}\hspace{-0.15mm}\left(\left(\hspace{-0.14mm}D\hspace{-0.12mm}-\hspace{-0.12mm}x\hspace{-0.13mm}\right)\hspace{-.12mm}^2\hspace{-0.12mm}+\hspace{-0.2mm}h\hspace{-.2mm}^2\right)\hspace{-0.0mm} \left(\frac{\hspace{-0.15mm}\eta_{_{\textrm{NLoS}}}}{\hspace{-0.15mm}\mu_{_{\textrm{NLoS}}}}\hspace{-.17mm}\right)}.$
      }
\end{equation}
Considering the conventional \textit{decode-and-forward} relay mode, the SIR of the system $\textrm{SIR}_S$ is given by:
 \begin{equation}\label{eq:SIR_s}
     \textrm{SIR}_S(x,h)=\min \big\{\textrm{SIR}_1(x,h),\textrm{SIR}_2(x,h)\big\} \;\;\forall x,h.
\end{equation}
Assuming equal bandwidths for both links, maximizing the data rate between the Tx and the Rx is equivalent to maximizing the $\textrm{SIR}_S$ by tuning the location of the UAV described as:
\begin{equation}
    (x^*,h^*) = \argmax_{x\in [0,D], h\in [h_{min},h_{max}]} \textrm{SIR}_S(x,h).
\end{equation}
The presence of an MSI renders our approach different from most of the works mentioned in Section~\ref{sec:intro} mainly due to its impact on the SIR expressions making them non-convex with respect to (w.r.t) the position of the UAV(s), which leads to the inapplicability of the conventional optimization techniques. In this work, we exploit \textit{geometry} and \textit{functional analysis} to obtain the subsequent derivations. In the following, we propose two lemmas, which are later used to derive the main results.
\begin{figure}[t]
\includegraphics[width=8.9cm,height=2.3cm]{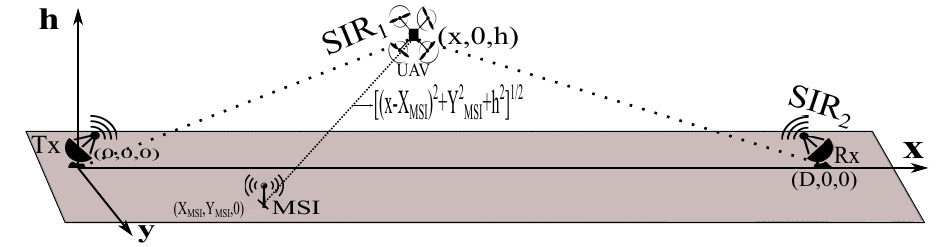}
		\caption{A single UAV acting as a relay between a pair of Tx and Rx coexisting with an MSI.}
		 \label{fig:single}
		 \vspace{-1.2mm}
		 \end{figure}
\begin{definition}
In geometry, a \textit{locus} is the set of all points satisfying the same conditions or possessing the same properties.
\end{definition}
\begin{lemma}\label{th:main}
The locus of the points satisfying $\textrm{SIR}_1(x,h)=\textrm{SIR}_2(x,h)$ is given by the following expression\footnote{In this work, the $+$ and $-$ superscripts always denote the larger and the smaller solution, respectively.}:
\begin{equation}\label{eq:locus}
    {h^{\pm}}=\sqrt{\Lambda^{\pm}(x)},  
\end{equation}
with $\Lambda^{\pm}(x)\triangleq \big[-B(x)\pm \sqrt{B^2(x)-4A(x)C(x)}\big]/\left(2A(x)\right)$, where $A(x)$, $B(x)$, and $C(x)$ are given by~\eqref{eq:lambda}.
\begin{table*}[t]
\begin{minipage}{0.99\textwidth}
\begin{equation}\label{eq:lambda}
\begin{aligned}
     &A(x)=p_t, B(x)= p_t\left(X_{_{\textrm{MSI}}}-x\right)^2+p_t(D-x)^2-p_u       \left(\frac{\mu_{_{\textrm{NLoS}}}}{\eta_{_{\textrm{NLoS}}}}\right)\left(D-X_{_{\textrm{MSI}}}\right)^2+Y_{_{\textrm{MSI}}}^2\left(p_t-p_u       \left(\frac{\mu_{_{\textrm{NLoS}}}}{\eta_{_{\textrm{NLoS}}}}\right)\right), \\
     &C(x)= p_t\big[(D-x)^2\left(\left(X_{_{\textrm{MSI}}}-x\right)^2+Y_{_{\textrm{MSI}}}^2\right) \big]- p_u       \left(\frac{\mu_{_{\textrm{NLoS}}}}{\eta_{_{\textrm{NLoS}}}}\right)\Big[x^2\left(Y_{_{\textrm{MSI}}}^2+ (D-X_{_{\textrm{MSI}}})^2\right) \Big]
     \end{aligned}
\end{equation}
\begin{equation}\label{eq:quartic}
\resizebox{0.95\hsize}{!}{$
p_t\left(x-X_{_{\textrm{MSI}}}\right)^2\left((D-x)^2+\hat{h}^2\right)+p_t( Y_{_{\textrm{MSI}}}^2+\hat{h}^2)(D-x)^2-p_u\left(\frac{\mu_{_{\textrm{NLoS}}}}{\eta_{_{\textrm{NLoS}}}}\right) x^2(Y_{_{\textrm{MSI}}}^2+(D-X_{_{\textrm{MSI}}})^2)+p_t \hat{h}^2(Y_{_{\textrm{MSI}}}^2+\hat{h}^2)-p_u\left(\frac{\mu_{_{\textrm{NLoS}}}}{\eta_{_{\textrm{NLoS}}}}\right) \hat{h}^2(Y_{_{\textrm{MSI}}}^2+(D-X_{_{\textrm{MSI}}})^2)=0$
    }
\end{equation}
\vspace{-10mm}
\end{minipage}
\end{table*}
\end{lemma}
\begin{proof}
The proof can be carried out using algebraic manipulations, which is omitted due to the limited space.
\end{proof}
\begin{lemma}\label{lemma:lem1}
For $\textrm{SIR}_1$, the \textit{stationary points}~\cite{ref:stat} with respect to $x$, $\Psi^x$, is given by: 
\begin{equation}
    \begin{aligned}
  &\Psi^x= \frac{Y_{_{\textrm{MSI}}}^2+X_{_{\textrm{MSI}}}^2+\sqrt{(Y_{_{\textrm{MSI}}}^2+X_{_{\textrm{MSI}}}^2)^2+4X_{_{\textrm{MSI}}}^2h^2}}{2X_{_{\textrm{MSI}}}}.
\end{aligned}
\end{equation}
Also, $\textrm{SIR}_1$ has no stationary point with respect to $h$ when $h\in (h_{min},h_{max})$. With $\Psi^h\triangleq \frac{Y_{_{\textrm{MSI}}}^2+X_{_{\textrm{MSI}}}^2}{2X_{_{\textrm{MSI}}}}$, we have
\begin{equation}
   \hspace{-.1mm} \begin{cases}
          \frac{\partial \textrm{SIR}_1(x,h)}{\partial x}\geq0 \;\textrm{if}\; x\geq \Psi^x,\\
          
          \frac{\partial \textrm{SIR}_1(x,h)}{\partial x}<0 \;\; \textrm{O.W.},
    \end{cases}
    \begin{cases}
          \frac{\partial \textrm{SIR}_1(x,h)}{\partial h}\geq0\;\textrm{if}\;x\geq \Psi^h,\\
          \frac{\partial \textrm{SIR}_1(x,h)}{\partial h}<0\;\; \textrm{O.W}.
    \end{cases}
\end{equation}
Moreover,
\begin{equation}
   \hspace{-9mm}\resizebox{0.99\hsize}{!}{ $\displaystyle\max_{\hspace{9mm}x\in [0,D], h\in [h_{min},h_{max}]} \hspace{-6mm} \textrm{SIR}_1(x,h)=\frac{p_t(X_{_{\textrm{MSI}}}^2+Y_{_{\textrm{MSI}}}^2+h_{min}^2)}{p_{_{\textrm{MSI}}} h_{min}^2}.$}
\end{equation}
On the other hand, $\textrm{SIR}_2$ has no stationary points when $x\in (0,D), h\in (h_{min},h_{max})$ and
\begin{align}
         & \hspace{-.5mm}\frac{\partial \textrm{SIR}_2(x,h)}{\partial h}\leq0, \;\;
          \frac{\partial \textrm{SIR}_2(x,h)}{\partial x}\geq 0,\nonumber\\
          &\forall x\in [0,D], h\in [h_{min},h_{max}],
          \end{align}
          and
          \begin{equation}
    \hspace{-10mm}\max_{\hspace{10mm} x\in [0,D], h\in [h_{min},h_{max}]}\hspace{-14mm} \textrm{SIR}_2(x,h) =\hspace{-0.5mm}\frac{p_u\mu_{_{\textrm{NLoS}}}(Y_{_{\textrm{MSI}}}^2+(D-X_{_{\textrm{MSI}}})^2)}{p_{_{\textrm{MSI}}}\eta_{_{\textrm{NLoS}}}h_{min}^2}.
\end{equation}
\end{lemma}

In practice, one of the following scenarios may occur: (i) The UAV position is vertically fixed and horizontally adjustable. This may arise in urban applications, in which there is a desired altitude for the UAVs to avoid collision with other flying objects. (ii) The UAV position is horizontally fixed and vertically adjustable.  This happens specially in the surveillance and information gathering applications, in which the position of the UAV is fixed in the desired horizontal position and only the altitude can be tuned. (iii) The UAV position is neither vertically nor horizontally fixed, which is practical in non-urban areas with a few flying objects. In the following, we tackle these scenarios in order. 
Henceforth, whenever we refer to the roots of an equation or the points in the locus, the feasible space is confined to $x\in[0,D]$ and $h\in[h_{min},h_{max}]$.

\subsubsection{Finding the optimal horizontal position $x^*$ of the UAV for a given altitude $h=\hat{h}$}
In this case, we analyze the result of Lemma~\ref{th:main} using Lemma~\ref{lemma:lem1} in the following Theorem. 
\begin{theorem}\label{cor:1}
Given a fixed altitude $h=\hat{h}$, the horizontal positions satisfying \eqref{eq:locus} can be obtained by solving the \textit{quartic} equation given in~\eqref{eq:quartic}, where the characteristic of this equation considering $x\in [0,D]$ and the corresponding optimal horizontal position of the UAV $x^*$ are described as follows:

\textbf{Case 1)} $\textrm{SIR}_1(0,\hat{h})< \textrm{SIR}_2(0,\hat{h})$: The quartic equation has no solution and $x^*=0$.

\textbf{Case 2)} $\textrm{SIR}_1(0,\hat{h})\geq \textrm{SIR}_2(0,\hat{h})$ and $\Psi^x\geq D$ and $\textrm{SIR}_2(D,\hat{h}) \geq \textrm{SIR}_1(D,\hat{h})$: The quartic equation has one solution $x_{sol}$, which can be numerically obtained and $x^*=x_{sol}$.
   
\textbf{Case 3)} $\textrm{SIR}_1(0,\hat{h})\geq \textrm{SIR}_2(0,\hat{h})$ and $\Psi^x\geq D$ and $\textrm{SIR}_2(D,\hat{h}) < \textrm{SIR}_1(D,\hat{h})$: The quartic equation has no solution and  $x^*=D$.
    
\textbf{Case 4)} $\textrm{SIR}_1(0,\hat{h})\geq \textrm{SIR}_2(0,\hat{h})$ and $\Psi^x<D$ and $\textrm{SIR}_1(\Psi^x,\hat{h})\leq \textrm{SIR}_2(\Psi^x,\hat{h})$: The quartic equation has at least a feasible solution. Let $x_{sol}$ denote the smallest solution, if $\textrm{SIR}_1(x_{sol},\hat{h})\geq \textrm{SIR}_1(D,\hat{h})$ then $x^*=x_{sol}$; otherwise, $x^*=D$.

      \textbf{Case 5)} $\textrm{SIR}_1(0,\hat{h})\geq \textrm{SIR}_2(0,\hat{h})$ and $\Psi^x<D$ and $\textrm{SIR}_1(\Psi^x,\hat{h})> \textrm{SIR}_2(\Psi^x,\hat{h})$: The quartic equation may or may not have a feasible solution and $x^*=D$.

\end{theorem}
\begin{proof}
The proof is the result of Lemma~\ref{lemma:lem1} considering the behavior of the SIR expressions. 
\end{proof}
\vspace{-1mm}
\subsubsection{Finding the optimal vertical position $h^*$ of the UAV for a given horizontal position $x=\hat{x}$}
In this case, the vertical positions (altitudes) satisfying~\eqref{eq:locus} can be easily derived since $\Lambda^{\pm}(x)$ on the right hand side of the equation is known. Using Lemma~\ref{lemma:lem1}, we obtain the following theorem.
\begin{theorem}\label{th:givenX}
Given a fixed horizontal position $x=\hat{x}$, the optimal altitude $h^*$ of the UAV is given by:

\textbf{Case 1)} $\hat{x}\leq \Psi^h$: $h^*=h_{min}$. 
 
\textbf{Case 2)} $\hat{x}> \Psi^h$ and~\eqref{eq:locus} has a feasible solution (either $h^+$ or $h^-$ belong to $[h_{min},h_{max}]$): $h^*$ is the same as the feasible solution of~\eqref{eq:locus}.
 
\textbf{Case 3)} $\hat{x}>\Psi^h$ and~\eqref{eq:locus} has no feasible solution: $h^*$ can be derived by solely inspecting the boundary positions:
     \begin{equation}
         h^*=\argmax_{h\in \{h_{min},h_{max}\}} \textrm{SIR}_S(\hat{x},h).
     \end{equation}
\end{theorem}
\begin{proof}
The proof is similar to the proof of Theorem~\ref{cor:1}.
\end{proof}
\subsubsection{Finding the optimal position when both $h$ and $x$ of the UAV are adjustable}
In the previous scenarios, the locus defined in~\eqref{eq:locus} reduces to an equation since one variable (either $h$ or $x$) is given, which is not the case here. In this case, the optimal position of the UAV is identified in the following theorem. 

\begin{theorem}\label{th:var}
Let $\Lambda$ denote the set of all the feasible solutions of the locus described in~\eqref{eq:locus}. The optimal position of the UAV $(x^*,h^*)$ is given by:

\textbf{Case \hspace{-.1mm}1)} If the Locus has no solution, the optimal position can be derived by solely examining the boundary positions:
  \begin{equation}
     (x^*,h^*)=\displaystyle \hspace{-15mm}\argmax_{(x,h)\in\{(0,h_{min}),(0,h_{max}),(D,h_{min}),(D,h_{max})\}}\hspace{-20mm} \textrm{SIR}_S(x,h).
    \end{equation}
    
\textbf{Case \hspace{-.1mm}2)} Upon having at least one feasible solution for the locus, if $\textrm{SIR}_2(D,h_{min})\leq \max \{ \textrm{SIR}_1(D,h_{max}),\textrm{SIR}_1(D,h_{min}) \}$, $(x^*,h^*)=(D,h_{min})$. Also, if $\textrm{SIR}_1(0,h_{min})\leq \textrm{SIR}_2(0,h_{min})$, $(x^*,h^*)=(0,h_{min})$. Otherwise, let $(\tilde{x},\tilde{h}) = \argmax_{(x,h)\in \Lambda} \textrm{SIR}_S(x,h)$, then the optimal position of the UAV is given as follows:
\begin{itemize}
      \item If $\psi^x\geq D$: $x^*=\tilde{x}$ and $h^*$ can be derived using Theorem~\ref{th:givenX} considering $\hat{x}=\tilde{x}$.
    \item  If $\psi^x< D$ and $\tilde{x} \geq \psi^x$: $x^*=D$ and $h^*$ can be derived using Theorem~\ref{th:givenX} considering $\hat{x}=D$.
       \item If $\psi^x<  D$ and $\psi^h\leq \tilde{x} <\psi^x$ and $\textrm{SIR}_1(\tilde{x},\tilde{h})\geq\textrm{SIR}_1(D,h_{max})$: $(x^*,h^*)=(\tilde{x},\tilde{h})$.
     \item If $\psi^x<  D$ and $\psi^h\leq \tilde{x} <\psi^x$ and $\textrm{SIR}_1(\tilde{x},\tilde{h})<\textrm{SIR}_1(D,h_{max})$: $x^*=D$ and $h^*$ can be derived using Theorem~\ref{th:givenX} considering $\hat{x}=D$.
        \item If $\psi^x< D$ and $\tilde{x} <\psi^h$ and  $\textrm{SIR}_1(\tilde{x},h_{min})\geq \textrm{SIR}_1(D,h_{max})$: $(x^*,h^*)=(\tilde{x},h_{min})$.
         \item If $\psi^x< D$ and $\tilde{x} <\psi^h$ and  $\textrm{SIR}_1(\tilde{x},h_{min})< \textrm{SIR}_1(D,h_{max})$: $x^*=D$ and $h^*$ can be derived using Theorem~\ref{th:givenX} considering $\hat{x}=D$.
      \end{itemize}
\end{theorem}
\begin{proof}
The proof is similar to the previous theorems.
\end{proof}
\vspace{-3mm}
\section{Position Planning for Multiple UAVs}\label{sec:multiple}
\noindent We investigate the placement planning upon utilizing
multiple UAVs from two different points of view. First, we consider a cost effective design, in which the network designer aims to identify the minimum required number of utilized UAVs and determine their positions so as to satisfy a predetermined SIR of the system. Second, we assume that the network designer is provided with a set of UAVs, and endeavors to configure their positions so as to maximize the SIR of the system. 

\begin{figure}[t]
\includegraphics[width=8.8cm,height=2.1cm]{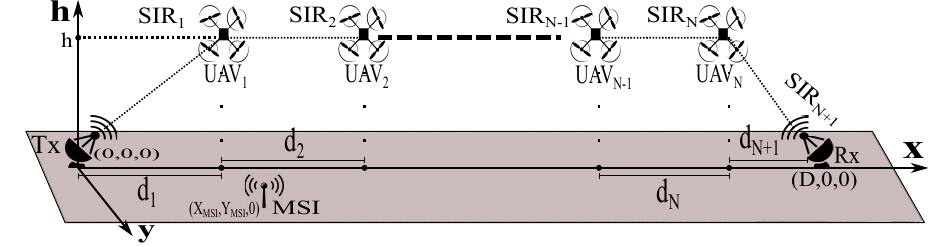}
		\caption{Multiple UAVs acting as relays between a pair of Tx and Rx coexisting with an MSI.}
		 \label{fig:multiple}
		  \vspace{-0.3mm}
		 \end{figure}
 \vspace{-2.0mm}
\subsection{Network Design to Achieve a Desired SIR}\label{subsec:1}
Let $\Gamma$ denote the desired SIR of the system and assume that $N$ is the minimum number of UAVs needed to satisfy the SIR constraint. We index the Tx node by $0$, the UAVs between the Tx and the Rx from $1$ to $N$, and the Rx node by $N+1$. We denote the horizontal distance between two consecutive nodes $i-1$ and $i$ by $d_i$, $ 1\leq i\leq N+1$, and consider $\mathbf{d}=[d_1,\cdots,d_{N+1}]$. To have tractable derivations, we assume that all the UAVs have the same altitude $h$. It can be verified that this assumption maximizes the SIR between two adjacent UAVs for a given horizontal distance. The model is depicted in Fig.~\ref{fig:multiple}. Let $\textrm{SIR}_k$ denote the SIR at the $k^{th}$ node, which can be obtained as:
\begin{equation}\label{eq:SIRs}
\hspace{-1mm}
\begin{aligned}
  & \textrm{SIR}_1(\mathbf{d},h)= \frac{p_t \eta^{-1}_{_{\textrm{NLoS}}}\left(\sqrt{d_1^2+h^2}\right)^{-2}}{p_{_{\textrm{MSI}}}\eta^{-1}_{_{\textrm{NLoS}}}\left(\sqrt{(X_{_{\textrm{MSI}}}-d_1)^2+Y^2_{\textrm{MSI}}+h^2}\right)^{-2}},\\ 
    \vdots\\\vspace{-1mm}
  & \textrm{SIR}_N(\mathbf{d},h)=\hspace{-0.5mm}\frac{p_u \mu^{-1}_{_{\textrm{LoS}}} \left(\sqrt{d_N^2}\right)^{-2}}{\hspace{-1.5mm}p_{_{\textrm{MSI}}}\eta^{-1}_{_{\textrm{NLoS}}}\hspace{-0.5mm}\left(\hspace{-0.5mm}\sqrt{\hspace{-0.5mm}(X_{_{\textrm{MSI}}}\hspace{-1.0mm}-\hspace{-0.5mm}\sum_{i=1}^{N} d_i)^2\hspace{-0.5mm}+\hspace{-0.5mm}Y^2_{\textrm{MSI}}\hspace{-0.5mm}+\hspace{-0.5mm}h^2\hspace{-0.5mm}}\right)^{\hspace{-0.5mm}-2}},\\ \vspace{-1mm}
  & \textrm{SIR}_{N+1}(\mathbf{d},h)=\frac{p_u\eta^{-1}_{_{\textrm{NLoS}}} \left(\sqrt{d_{N+1}^2+h^2}\right)^{-2}}{p_{_{\textrm{MSI}}}\mu^{-1}_{_{\textrm{NLoS}}}\left(\sqrt{(X_{_{\textrm{MSI}}}-D)^2+Y^2_{\textrm{MSI}}}\right)^{-2}}.
    \end{aligned}
 \end{equation} 
Similar to the single UAV scenario, $\textrm{SIR}_S $ is given by:
 \begin{equation}\label{eq:SIR_s2}
    \textrm{SIR}_S=\min \big\{\textrm{SIR}_1,\textrm{SIR}_2,\cdots, \textrm{SIR}_N, \textrm{SIR}_{N+1}\big\}.
\end{equation}
\subsubsection{The SIR expressions and the feasibility constraints}
From~\eqref{eq:SIRs}, it can be observed that achieving any desired $\textrm{SIR}_S$ ($\Gamma$) may not be feasible. To derive the feasibility conditions for the $\Gamma$, we need to analyze the links between the Tx to $\textrm{UAV}_1$, among the adjacent UAVs, and from $\textrm{UAV}_{N}$ to the Rx.

Analysis of the links between the Tx and $\textrm{UAV}_1$ ($\textrm{SIR}_1$) and between $\textrm{UAV}_{N}$ and the Rx ($\textrm{SIR}_{N+1}$) is similar to the discussion provided in Section~\ref{sec:singleUAV} (see Lemma~\ref{lemma:lem1}). Hence, we skip them and consider the SIR at $\textrm{UAV}_i$, $ 2\leq i\leq N$. For this UAV, the stationary point $\Phi_i^d$ of the SIR expression is given by:
\begin{equation}
    \Phi_i^d=\frac{h^2+Y_{_{\textrm{MSI}}}^2+(X_{_{\textrm{MSI}}}-\sum_{j=1}^{i-1} d_j)^2}{X_{_{\textrm{MSI}}}-\sum_{j=1}^{i-1} d_j},
\end{equation}
using which it can be validated that:
\begin{equation}\label{eq:dmiiidds}
\resizebox{0.89\hsize}{!}{$
\hspace{-15.5mm}
\displaystyle\max_{\hspace{13mm} x\in [0,D], h\in [h_{min},h_{max}]} \hspace{-11.5mm}\textrm{SIR}_i(x,h) \hspace{-.5mm}\leq\hspace{-.5mm} \frac{p_u\mu^{-1}_{_{\textrm{LoS}}}\hspace{-.5mm}\left(\max(X_{_{\textrm{MSI}}},D-X_{_{\textrm{MSI}}})+ Y_{_{\textrm{MSI}}}^2\hspace{-.5mm}+\hspace{-.5mm}h^2\right)}{p_{_{\textrm{MSI}}}\eta^{-1}_{_{\textrm{NLoS}}} d^2_{min} },\;$}
\end{equation}
\begin{equation}
\begin{cases}
    \frac{\partial \textrm{SIR}_i(x,h)}{\partial d_i}\hspace{-.5mm}\geq \hspace{-.5mm}0 \; \textrm{if} \; d_i\hspace{-.5mm}\geq\hspace{-.5mm}\Phi_i^d,\\
    \frac{\partial \textrm{SIR}_i(x,h)}{\partial d_i}\hspace{-.5mm}< \hspace{-.5mm}0 \; \textrm{O.W.},
    \end{cases}
\end{equation}
where $d_{min}$ is the minimum feasible distance between two UAVs considering the mechanical constrains. Combining these derivations with those in Section~\ref{sec:singleUAV}, we obtain the feasibility condition declared in~\eqref{eq:conditioforSIR}.
\begin{table*}[t]
\vspace{1.3mm}
\begin{minipage}{0.99\textwidth}
\begin{equation}\label{eq:conditioforSIR}
    \Gamma \leq \min \Bigg\{\frac{p_t\left(X^2_{\textrm{MSI}}+Y_MSI^2+h^2 \right)}{p_{_{\textrm{MSI}}} h^2},\frac{p_u\mu^{-1}_{_{\textrm{LoS}}}\hspace{-.5mm}\left(\max(X_{_{\textrm{MSI}}},D-X_{_{\textrm{MSI}}})+ Y_{_{\textrm{MSI}}}^2\hspace{-.5mm}+\hspace{-.5mm}h^2\right)}{p_{_{\textrm{MSI}}}\eta^{-1}_{_{\textrm{NLoS}}} d^2_{min} } , \frac{p_u\eta^{-1}_{_{\textrm{NLoS}}}\left(Y_{_{\textrm{MSI}}}^2+(D-X_{_{\textrm{MSI}}})^2\right)}{p_{_{\textrm{MSI}}}\mu^{-1}_{_{\textrm{NLoS}}}h^2}\Bigg\}
\end{equation}
\end{minipage}
\end{table*}
\subsubsection{Design Methodology}
To derive the minimum number of needed UAVs and their optimal positions so as to satisfy a desired $\textrm{SIR}_S$, we pursue the following three main steps: (i) Considering $\textrm{SIR}_1$, for $\textrm{UAV}_1$, we obtain the maximum distance from the Tx (toward the Rx) ${d}^*_1$ that satisfies the SIR constraint. (ii) Considering $\textrm{SIR}_{N+1}$, for $\textrm{UAV}_{N}$, we obtain the maximum distance from the Rx (toward the Tx)  ${d}^*_{N+1}$ that satisfies the desired SIR. (iii) Consider the segment between $\textrm{UAV}_1$ and $\textrm{UAV}_{N}$ with length $D-{d}^*_1-{d}^*_{N+1}$, we use the SIR expressions for the remaining UAVs to minimize the number of UAVs required to cover the distance while satisfying the desired $\textrm{SIR}_S$. In the following, we explain these steps in more detail.

Considering $\textrm{SIR}_1$, we solve $\textrm{SIR}_1=\Gamma$, the answer of which is given by~\eqref{eq:d_1}. \begin{table*}[t]
\begin{minipage}{0.99\textwidth}
\vspace{-2.8mm}
\begin{equation}\label{eq:d_1}
 d^+_1,d^-_1,=\frac{p_t X_{_{\textrm{MSI}}}\pm\sqrt{p^2_t X_{_{\textrm{MSI}}}^2-(p_t-\Gamma p_{_{\textrm{MSI}}})\left(p_t\left(X_{_{\textrm{MSI}}}^2+Y_{_{\textrm{MSI}}}^2\right)+h^2(p_t-p_{_{\textrm{MSI}}}\Gamma)\right)}}{p_t-\Gamma p_{_{\textrm{MSI}}}}
\end{equation}
\end{minipage}
\end{table*}
Then, using Lemma~\ref{lemma:lem1}, $d^*_1$ is given by:
\begin{equation}\label{eq:D1}
    d^*_1=\begin{cases}
              d_1^- \;\; \textrm{if}\;\; d_1^+> D,\\
            d_1^+ \;\; \textrm{if}\;\; d_1^+< \Psi^x \textrm{ and } d_1^+\leq D,\\
            D \;\; \textrm{O.W}.
    \end{cases}
\end{equation}
In the last case of \eqref{eq:D1}, the optimal number of UAVs is $1$, and the UAV should be placed at $x=D$. Assuming $d^*_1< D$, using Lemma~\ref{lemma:lem1}, $d^*_{N+1}$ can be obtained as:
\begin{equation}\label{eq:DNplus1}
  d^*_{N+1}=\sqrt{\frac{p_u\eta^{-1}_{_{\textrm{NLoS}}}\left((X_{_{\textrm{MSI}}}-D)^2+Y_{_{\textrm{MSI}}}^2\right)}{\Gamma p_{_{\textrm{MSI}}}\mu^{-1}_{_{\textrm{NLoS}}}}-h^2}.
\end{equation}
Afterward, we solve $\textrm{SIR}_{k}=\Gamma$ and use \eqref{eq:dmiiidds} to obtain $d^*_k$, $2\leq k \leq N$, which can be obtained as:
\begin{equation}\label{eq:ds}
 d^*_k=\begin{cases}
      d_k^- \;\; \textrm{if}\;\; d_k^+> D-\sum_{j=1}^{k-1} d^*_j-d^*_{N+1},\\
      d_k^+ \;\; \textrm{if}\;\; d_k^+< \Phi_k^d \textrm{ and } d_k^+\leq D-\sum_{j=1}^{k-1} d^*_j-d^*_{N+1},\\
        D-d^*_{N+1}-d_{min} \;\; \textrm{O.W.},
  \end{cases}
\end{equation}
where $d^-_k,d^+_k$ are given in~\eqref{eq:d_k}.
\begin{table*}[t]
\vspace{-3.3mm}
\begin{minipage}{0.99\textwidth}
\begin{equation}\label{eq:d_k}
\resizebox{0.95\hsize}{!}{$
          d^+_k,d^-_k=\frac{p_u \mu^{-1}_{_{\textrm{LoS}}}\left(X_{_{\textrm{MSI}}}-\sum_{i=1}^{k-1}d^*_i \right)\pm\sqrt{p^2_u \mu^{-2}_{\textrm{LoS}}\left(X_{_{\textrm{MSI}}}-\sum_{i=1}^{k-1} d^*_i \right)^2-p_u \mu^{-1}_{_{\textrm{LoS}}}\left(p_u \mu^{-1}_{_{\textrm{LoS}}}-\Gamma p_{_{\textrm{MSI}}}\eta^{-1}_{_{\textrm{NLoS}}}\right)\left(h^2+Y_{_{\textrm{MSI}}}^2+\left(X_{_{\textrm{MSI}}}-\sum_{i=1}^{k-1} d^*_i \right)^2\right)}}{p_u \mu^{-1}_{_{\textrm{LoS}}}-\Gamma p_{_{\textrm{MSI}}}\eta^{-1}_{_{\textrm{NLoS}}}}$}
\end{equation}
\vspace{-6.3mm}
\end{minipage}
\end{table*}
Finally, the minimum number of required UAVs $N^*$ is given by:
\begin{equation}\label{eq:Opt}
    N^*=\argmin_{N\in\mathbb{N}} \sum _{k=2}^{N} d^*_k\geq D-d^*_1-d^*_{N+1}.
\end{equation}
Note that according to~\eqref{eq:ds} and~\eqref{eq:d_k}, calculation of each $d^*_k$ only requires the knowledge of $d^*_{k'}$, $\forall k' < k$. Hence, the solution of \eqref{eq:Opt} can be easily obtained by initially assuming $N=2$ and increasing the value of $N$ by $1$ until the constraint in the right hand side of the equation is met.
\vspace{-1mm}
\subsection{Position Planning for a Given Number of UAVs}
In this case, there exist multiple UAVs dedicated as relays to the network, which are expected to be positioned to maximize the SIR of the system. To this end, 
an algorithm can be immediately proposed based on our results in the previous subsection, which considers the number of UAVs as given and gradually increases $\Gamma$ starting from $\Gamma=0$ to find the maximum value of $\Gamma$ for which $N^*$  in~\eqref{eq:Opt} becomes equal to the number of given UAVs. Afterward, the positions of the UAVs can be obtained as discussed before. Nevertheless, this is a centralized approach. In the following, we propose a distributed algorithm for the same purpose, where the UAVs locally compute their positions based on the knowledge of the positions of their adjacent neighbors, which can be obtained through simple message passing. Considering the SIR expressions in~\eqref{eq:SIRs}, for $\textrm{UAV}_{\floor{N/2}+2}$ to $\textrm{UAV}_{\floor{N}}$, we express the SIRs w.r.t the positions of the UAVs located after them (closer to Rx) as:
\begin{equation}\label{eq:SIR_multi}
\hspace{-0mm}
\begin{aligned}
  & \displaystyle \textrm{SIR}_{\floor{N/2}+2}(\mathbf{d},h)\hspace{-.5mm}=\hspace{-.5mm}\frac{\hspace{-4mm}p_u\mu^{-1}_{_{\textrm{LoS}}}\hspace{-.5mm}\left(\hspace{-.5mm}\left(\hspace{-.5mm}X_{_{\textrm{MSI}}}-D+\hspace{-5mm}\displaystyle\sum_{i=\floor{N/2}+3}^{N+1}\hspace{-4mm}d_{i}\right)^2\hspace{-.5mm}\hspace{-.5mm}\hspace{-.5mm}+\hspace{-1mm}Y_{_{\textrm{MSI}}}^2\hspace{-1mm}+h^2\hspace{-1mm}\right)}{p_{_{\textrm{MSI}}}\eta^{-1}_{_{\textrm{NLoS}}}d_{\floor{N/2}+2}^2},\\
  & \textrm{SIR}_{\floor{N/2}+3}(\mathbf{d},h)\hspace{-.5mm}=\hspace{-.5mm}\frac{\hspace{-4mm}p_u\mu^{-1}_{_{\textrm{LoS}}}\left(\hspace{-.5mm}\left(\hspace{-.5mm}X_{_{\textrm{MSI}}}-D+\hspace{-5mm}\displaystyle\sum_{i=\floor{N/2}+4}^{N+1}\hspace{-4mm}d_{i}\right)^2\hspace{-.5mm}\hspace{-.5mm}\hspace{-.5mm}+\hspace{-1mm}Y_{_{\textrm{MSI}}}^2\hspace{-1mm}+h^2\hspace{-1mm}\right)}{p_{_{\textrm{MSI}}}\eta^{-1}_{_{\textrm{NLoS}}}d_{\floor{N/2}+3}^2},\\
  \vdots\\
  &\textrm{SIR}_{N}(\mathbf{d},h)\hspace{-0.5mm}=\hspace{-.5mm}\frac{p_u\mu^{-1}_{_{\textrm{LoS}}}\left(\hspace{-.5mm}\left(X_{_{\textrm{MSI}}}-D+d_{N+1}\right)^2\hspace{-.5mm}+Y_{_{\textrm{MSI}}}^2+h^2\hspace{-.5mm}\right)}{p_{_{\textrm{MSI}}}\eta^{-1}_{_{\textrm{NLoS}}}d_N^2}.
\end{aligned}
\end{equation}
The following facts are immediate consequences of examining \eqref{eq:SIR_multi} and \eqref{eq:d_k}:  (i) With a known $d_{N+1}$ and a (hypothetically) given value for $\textrm{SIR}_S$ ($\Gamma$), starting with $\textrm{UAV}_{N}$ the distance between the subsequent UAVs can be locally obtained up to $\textrm{UAV}_{\floor{N/2}+2}$ using a \textit{backward propagation}, by which each UAV transmits its position rearward to the adjacent UAV located toward the Tx (see \eqref{eq:SIR_multi}), where
\begin{equation}\label{eq:Dmid}
    \hspace{0mm}d_k=\sqrt{\hspace{-1mm}\frac{\left(X_{_{\textrm{MSI}}}-D+\hspace{-2mm}\displaystyle\sum_{i=k+1}^{N+1}\hspace{-2mm}d_{i}\right)^2\hspace{-2mm}+Y_{_{\textrm{MSI}}}^2+h^2}{\Gamma p_{_{\textrm{MSI}}}\mu_{_{\textrm{LoS}}} /\left(p_u\eta_{_{\textrm{NLoS}}}\right)}}, \;\; \floor{N/2}+2\leq k\leq N. 
\end{equation}
(ii) With a known $d_1$ and a (hypothetically) given value for the $\textrm{SIR}_S$ ($\Gamma$), starting with $\textrm{UAV}_1$, the distance between the subsequent UAVs can be obtained up to $\textrm{UAV}_{\floor{N/2}}$ in a \textit{forward propagation}, by which each UAV transmits its position to the adjacent UAV located toward the Rx (see~\eqref{eq:ds}, \eqref{eq:d_k}). Hence, to obtain the positions three parameters are needed: $d_1$, $d_{N+1}$ and $\Gamma$. Note that in the mentioned propagations, no message is exchanged between the two UAVs in the middle ($\textrm{UAV}_{\floor{N/2}},\textrm{UAV}_{\floor{N/2}+1}$), and thus the SIR at $\textrm{UAV}_{\floor{N/2}+1}$ might be less than $\Gamma$. Given these facts, we propose a distributed algorithm for position planning of multiple UAVs, the pseudo code of which is given in Algorithm~\ref{alg:fulldist}. In this algorithm, we first derive the position of the first UAV, the position of the last UAV, and the initial desired $\textrm{SIR}_S$ ($\Gamma$) (lines~\ref{lineinAlg1}-\ref{lineinAlg2}). Afterward, starting from the beginning and the end, the UAVs locally obtain their positions w.r.t the position of their adjacent UAVs (lines~\ref{lineinAlg3}-\ref{lineinAlg4}) so as to satisfy the desired $\textrm{SIR}_S$. Then, the SIR at $\textrm{UAV}_{\floor{N/2}+1}$ is inspected (line~\ref{lineinAlg5}). If this SIR satisfies the desired SIR of the system ($\textrm{SIR}_{\floor{N/2}+1}\geq \Gamma$), the algorithm stops; otherwise, it moves the first and the last UAVs and starts over with a new desired value for $\Gamma$ (lines~\ref{lineinAlg6}-\ref{lineinAlg7}). Note that simultaneous identification of the positions achieved through using forward and backward propagations leads to a faster convergence since at each time instant two distances are calculated in parallel.

 \begin{algorithm}[t]
 	\caption{\makebox[7cm]{\small Distributed position planning for multiple UAVs}}\label{alg:fulldist}
 	\SetKwFunction{Union}{Union}\SetKwFunction{FindCompress}{FindCompress}
 	\SetKwInOut{Input}{input}\SetKwInOut{Output}{output}
 	 	{\footnotesize
 	\Input{Horizontal step size $\epsilon$.}
  $i=0$,
  $d^{(i)}_{1}=0,d^{(i)}_{N+1}=0	$.\\
  $\Gamma=\min \big\{\textrm{SIR}_1(d^{(i)}_{1},h),\textrm{SIR}_{N+1}(d^{(i)}_{N+1},h)\big\}$.\label{lineinAlg1}\\	
  \uIf{$\Gamma=\textrm{SIR}_1(d^{(i)}_{1},h)$}{
  \hspace{-3.5mm}\resizebox{0.350\textwidth}{!}{ Place $\textrm{UAV}_1$ at $x=d^{(i)}_{1}$ and re-derive $d^{(i)}_{N+1}$ using~\eqref{eq:DNplus1}}. 
  }\Else{
  \hspace{-3.5mm}\resizebox{0.385\textwidth}{!}{ Place $\textrm{UAV}_2$ at $x=D-{d^{(i)}_{N+1}}$ and re-derive $d^{(i)}_{1}$ using~\eqref{eq:D1}.}\\
}\label{lineinAlg2}
  Given $d^{(i)}_{N+1}$, obtain $d^{(i)}_{N},d^{(i)}_{N-1},\cdots,d^{(i)}_{\floor{N/2}+2}$ using backward propagation based on~\eqref{eq:Dmid}.\label{lineinAlg3}\\
  Given $d^{(i)}_{1}$, obtain $d^{(i)}_{2},\cdots,d^{(i)}_{\floor{N/2}}$ using forward propagation based on~\eqref{eq:ds}.\label{lineinAlg4}\\
 \resizebox{0.465\textwidth}{!}{ \hspace{-1mm}Send a message from $\textrm{UAV}_{\floor{N/2}}$ to $\textrm{\footnotesize UAV}_{\floor{N/2}+1}$ and measure ${\footnotesize \textrm{SIR}_{\floor{N/2}+1}}$.\label{lineinAlg5}}\\
  \uIf{SIR$_{\floor{N/2}+1}<\Gamma$}{\label{lineinAlg6}
  $d^{(i+1)}_{1}=d^{(i)}_{1}+\epsilon$\\
  $d^{(i+1)}_{N+1}=d^{(i)}_{N+1}+\epsilon$\\
  $i=i+1$ and go to line~\ref{lineinAlg1}.\label{lineinAlg7}
  }\Else{
  Fix the UAVs at their current positions.
  }
  }
  \end{algorithm}
  \vspace{-2mm}
  \section{Numerical Results}
  \noindent Similar to~\cite{chen2018multiple}, we consider $f_c=2\textrm{GHz}$, $C_{_{\textrm{LoS}}}=10^{0.01}$, $C_{_{\textrm{NLoS}}}=10^{2.1}$, and $\eta_{_{\textrm{NLoS}}}=\mu_{_{\textrm{NLoS}}}$. Assuming $p_u=1\textrm{W}$ and $D=35\textrm{m}$, Fig.~\ref{fig:locus} depicts the locus described in Lemma~\ref{th:main} for various parameters. Considering the solid black line and the dotted red line as the references, as expected, increasing $p_t$ (the marked blue and dashed magenta lines) shifts the locus toward the Rx. Considering the dotted red line and the dashed magenta line as the references, bringing the MSI closer to the Tx/Rx (the solid black and marked blue lines) shifts the locus downward, which is equivalent to a decrease in the required UAV altitude. 
  
  Fig.~\ref{fig:rate} compares the $\textrm{SIR}$ of the system obtained using Theorem~\ref{cor:1} to both the random placement, the performance of which is obtained by randomly placing the UAV in $1000$ Monte-Carlo iterations, and the method described in~\cite{chen2018multiple}, which does not capture the existence of the MSI (see Section~\ref{sec:intro}). In this simulation, it is assumed that $Y_{_{MSI}}=X_{_{MSI}}=30\textrm{m},D=35\textrm{m}, h\in[10\textrm{m},50\textrm{m}]$, and the simulation is conducted for two realizations of $p_t$ and $p_u$ described in the figure. As can be seen, the difference between the performance of our approach and the baselines is more prominent in low altitudes (up to $65\%$ increase in $\textrm{SIR}_S$). Also, on average, our method leads to $30.14\%$ and $25.73\%$ increase in $\textrm{SIR}_S$ as compared to the random placement and the method of~\cite{chen2018multiple}, respectively. 
    \begin{figure}[t]
	\minipage{4.2cm}
		\includegraphics[width=1\linewidth, height=.92\linewidth]{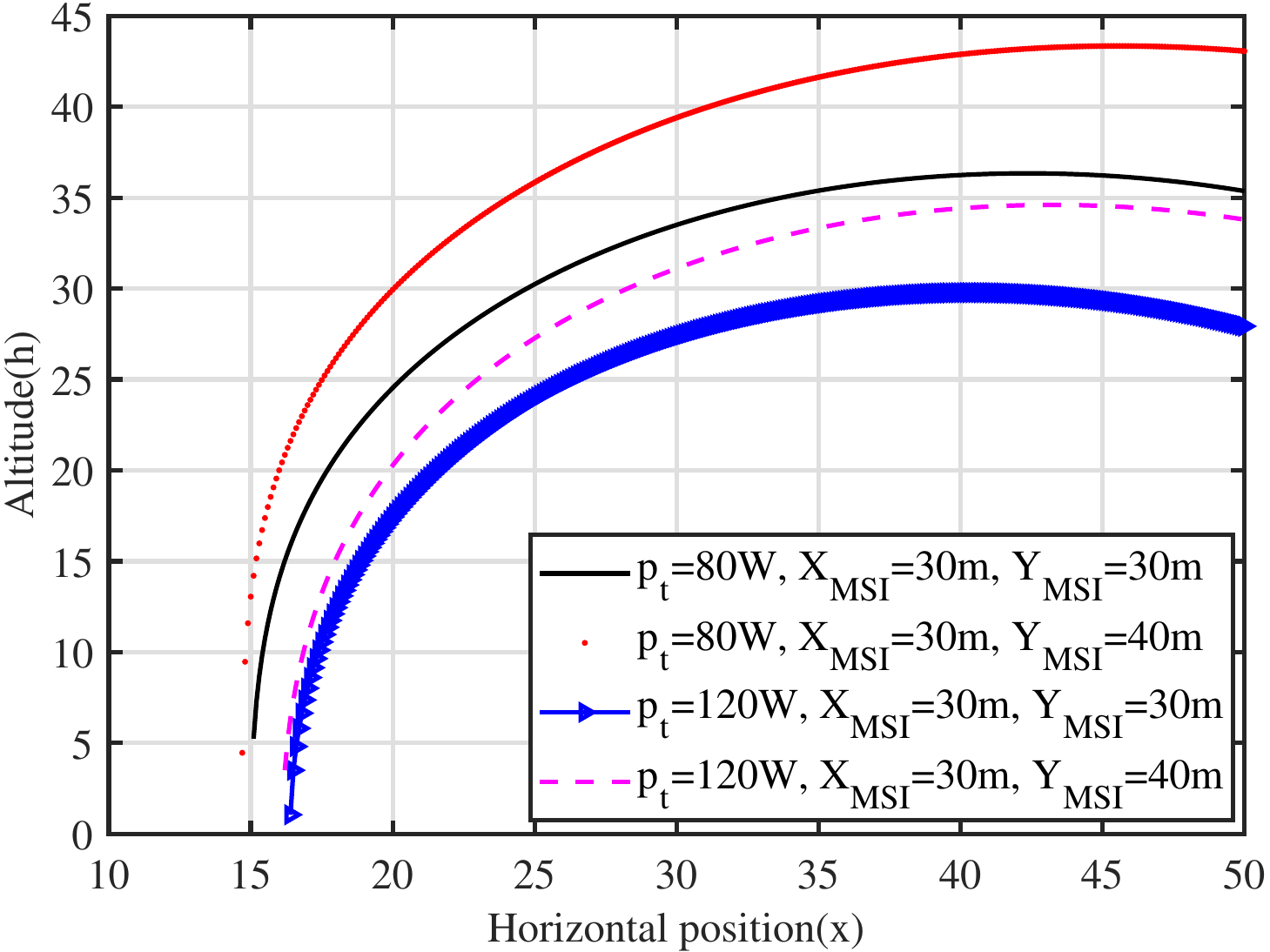}
		\caption{The locus of the points described in Lemma~\ref{th:main} for different parameters. \label{fig:locus}}
		\endminipage
		\quad
	\minipage{4.2cm}
		\includegraphics[width=1.0\linewidth, height=.901\linewidth]{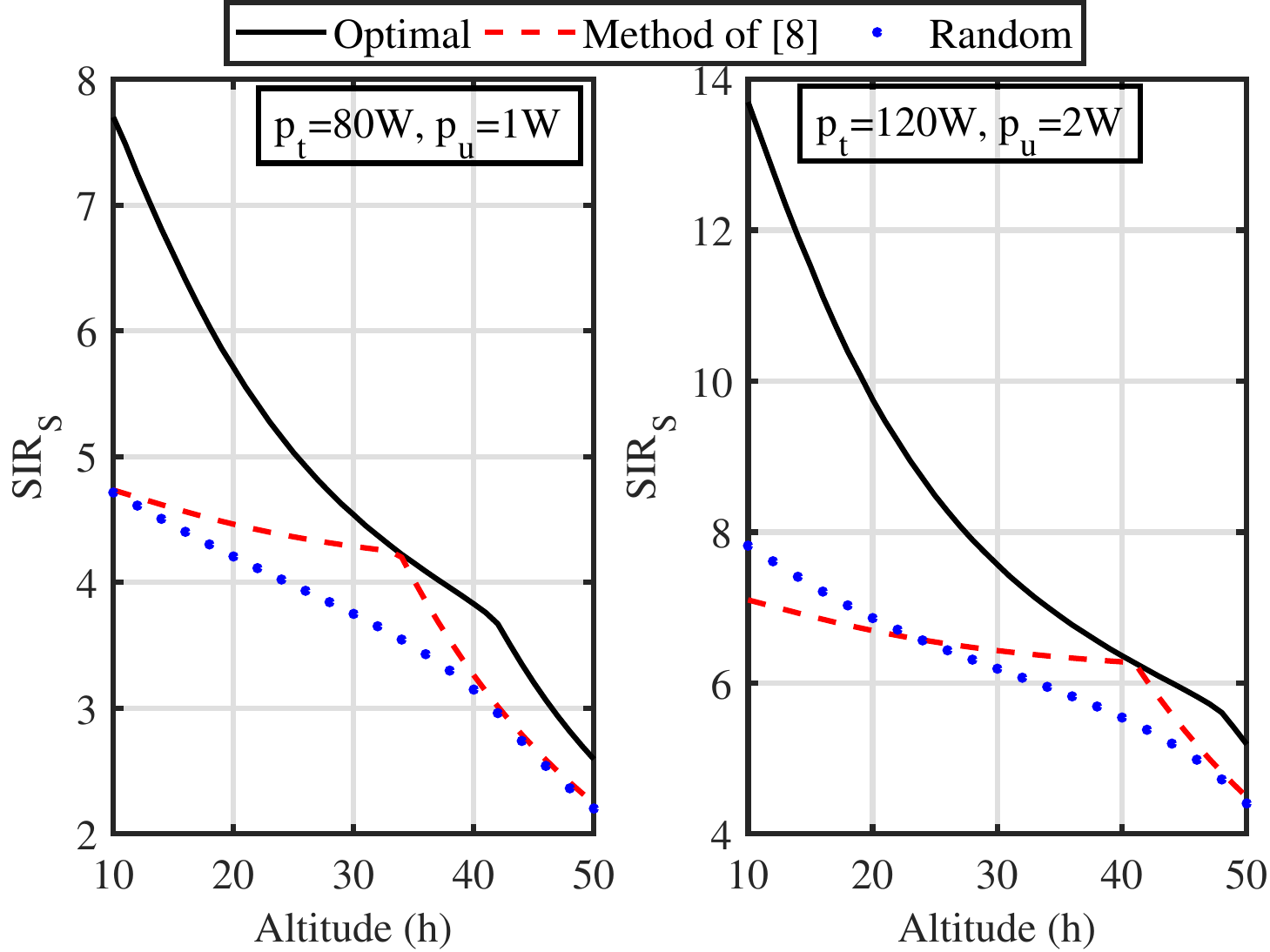}
		\caption{Comparison between the $\textrm{SIR}_S$ obtained using our optimal approach, the method in~\cite{chen2018multiple}, and the random placement for different parameters considering a single UAV.\label{fig:rate}}
		\endminipage

	\minipage{4.2cm}
		\includegraphics[width=1\linewidth, height=.901\linewidth]{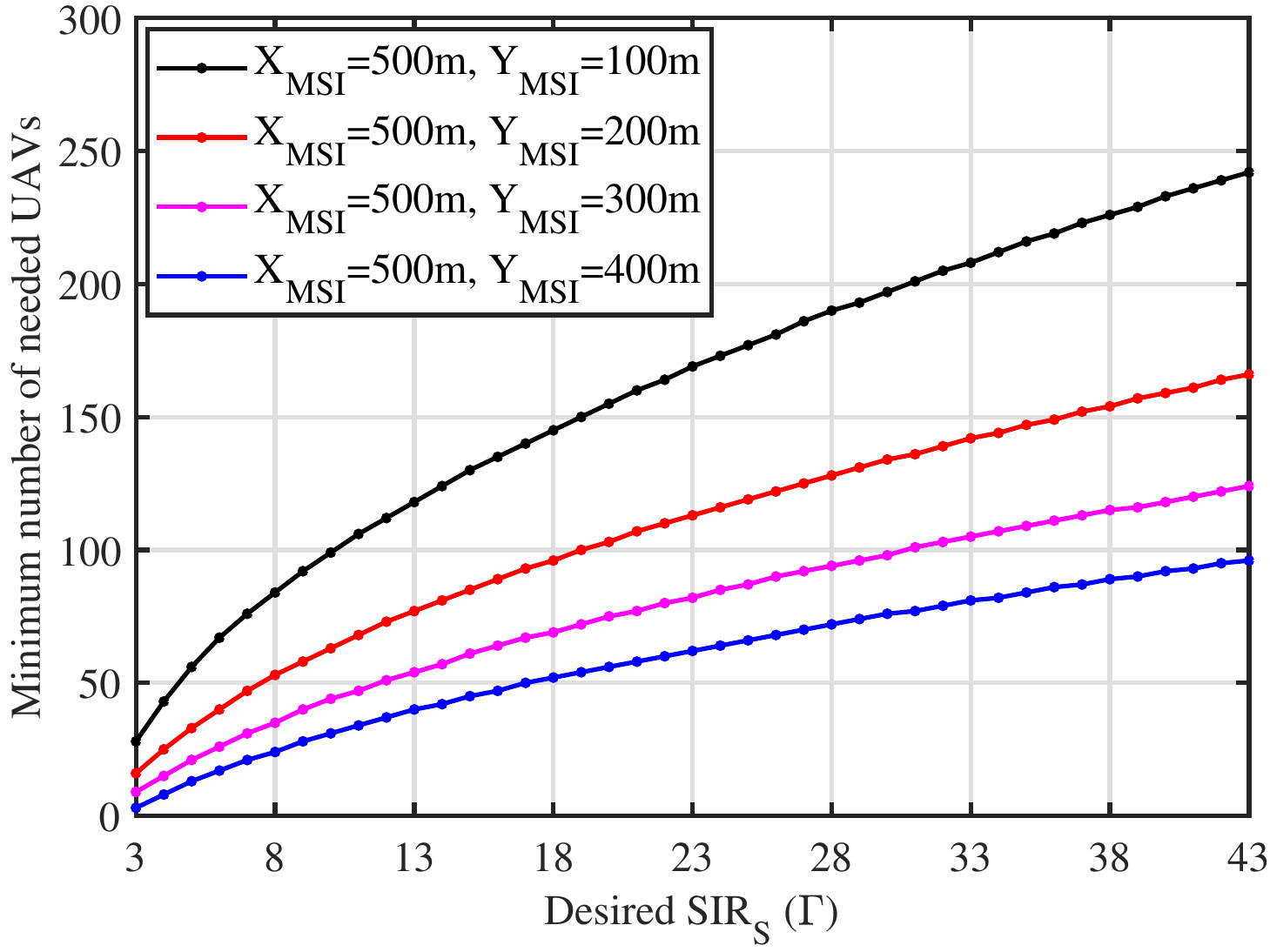}
		\caption{The minimum required number of UAVs to satisfy various values of $\textrm{SIR}_S$ for different positions of the MSI. \label{fig:minpos}}
		\endminipage
		\quad
	\minipage{4.2cm}
		\includegraphics[width=1\linewidth, height=.901\linewidth]{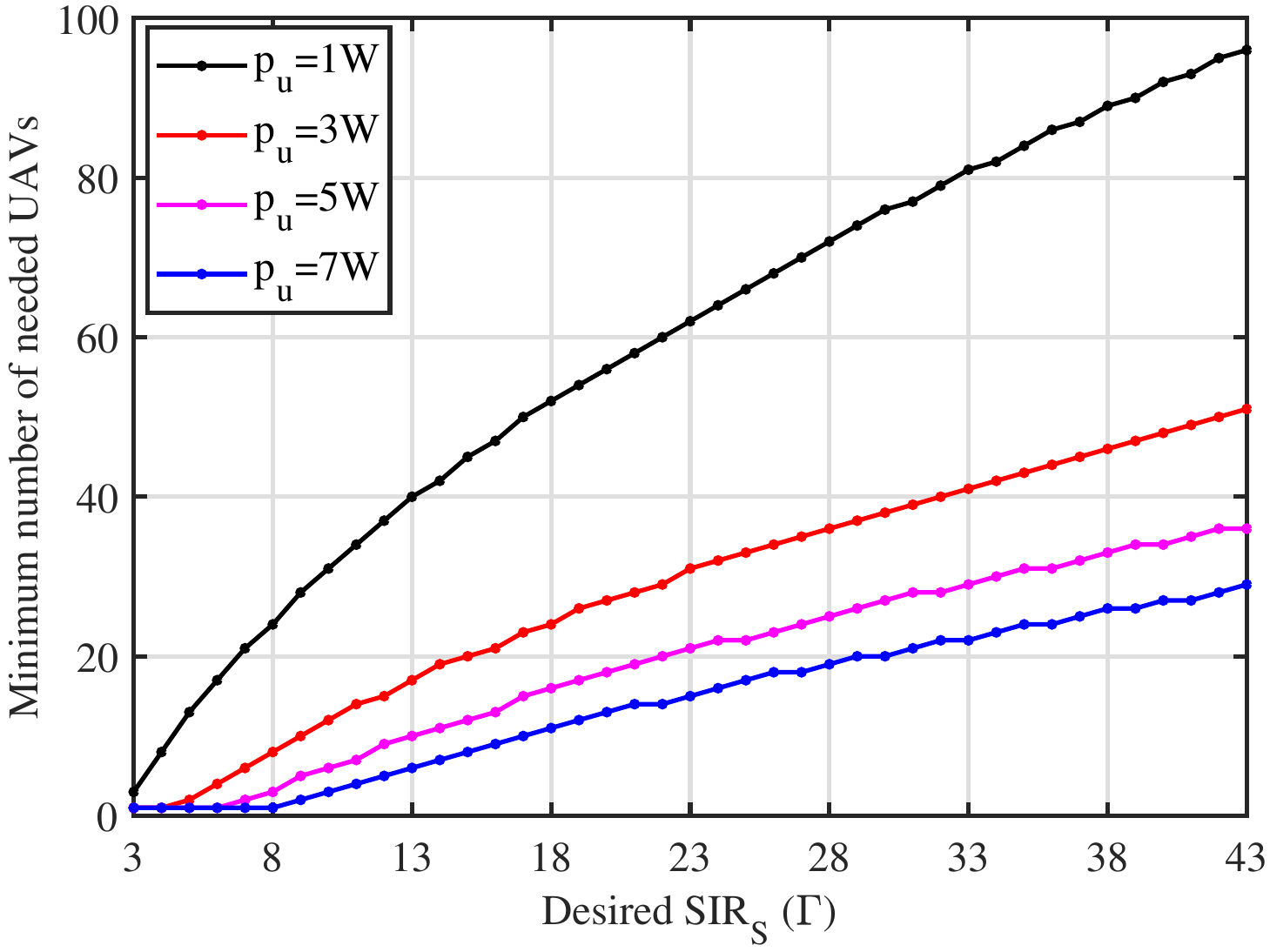}
		\caption{The minimum required number of UAVs to satisfy various values of $\textrm{SIR}_S$ for different UAV transmission powers.\label{fig:minpower}}
		\endminipage
		
	\minipage{4.2cm}
		\includegraphics[width=1\linewidth, height=.901\linewidth]{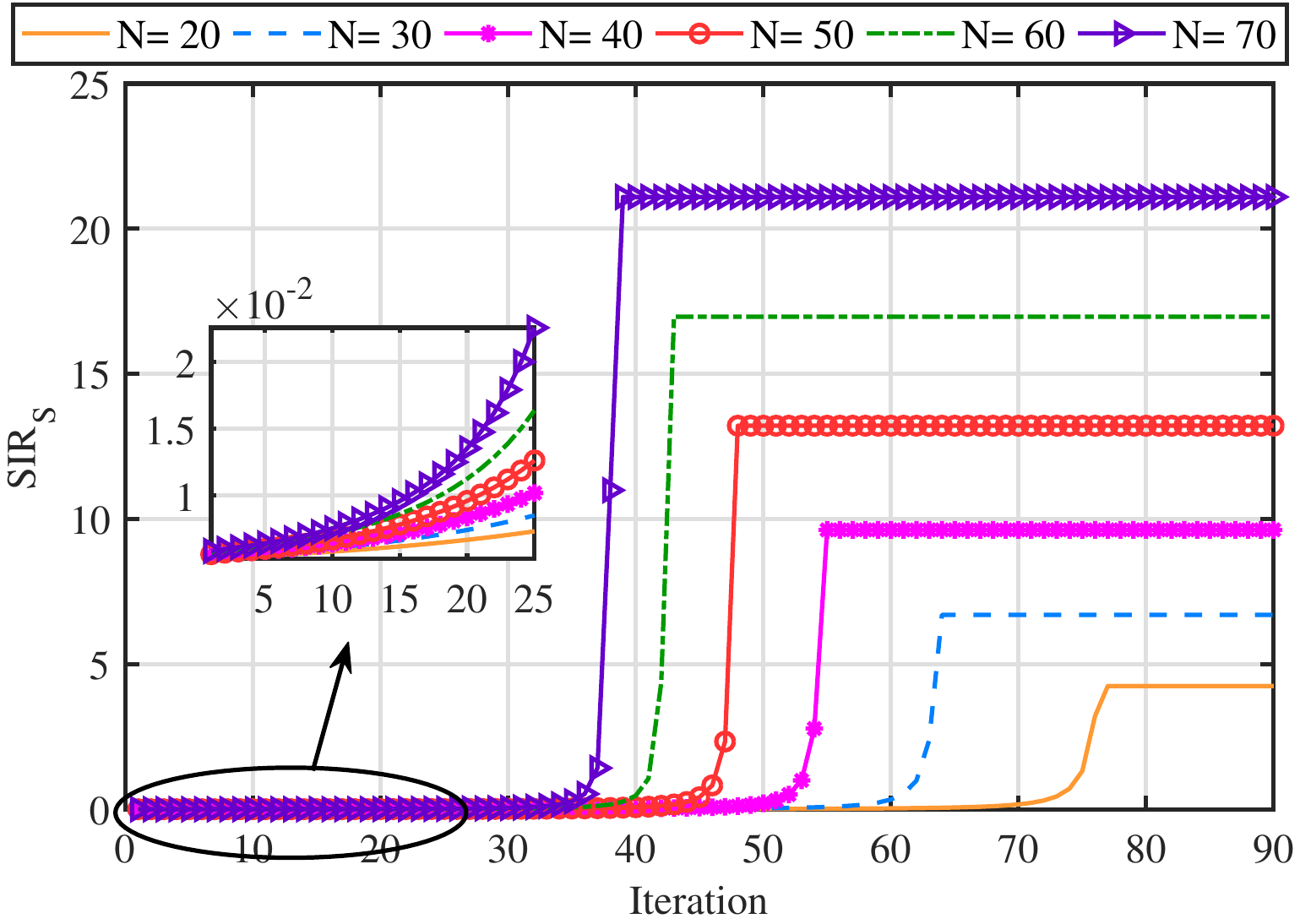}
		\caption{The $\textrm{SIR}_S$ w.r.t the iteration for our proposed distributed algorithm considering various number of UAVs in the system. \label{fig:conv}}
		\endminipage
			\quad
	\minipage{4.2cm}
		\includegraphics[width=1\linewidth, height=0.901\linewidth]{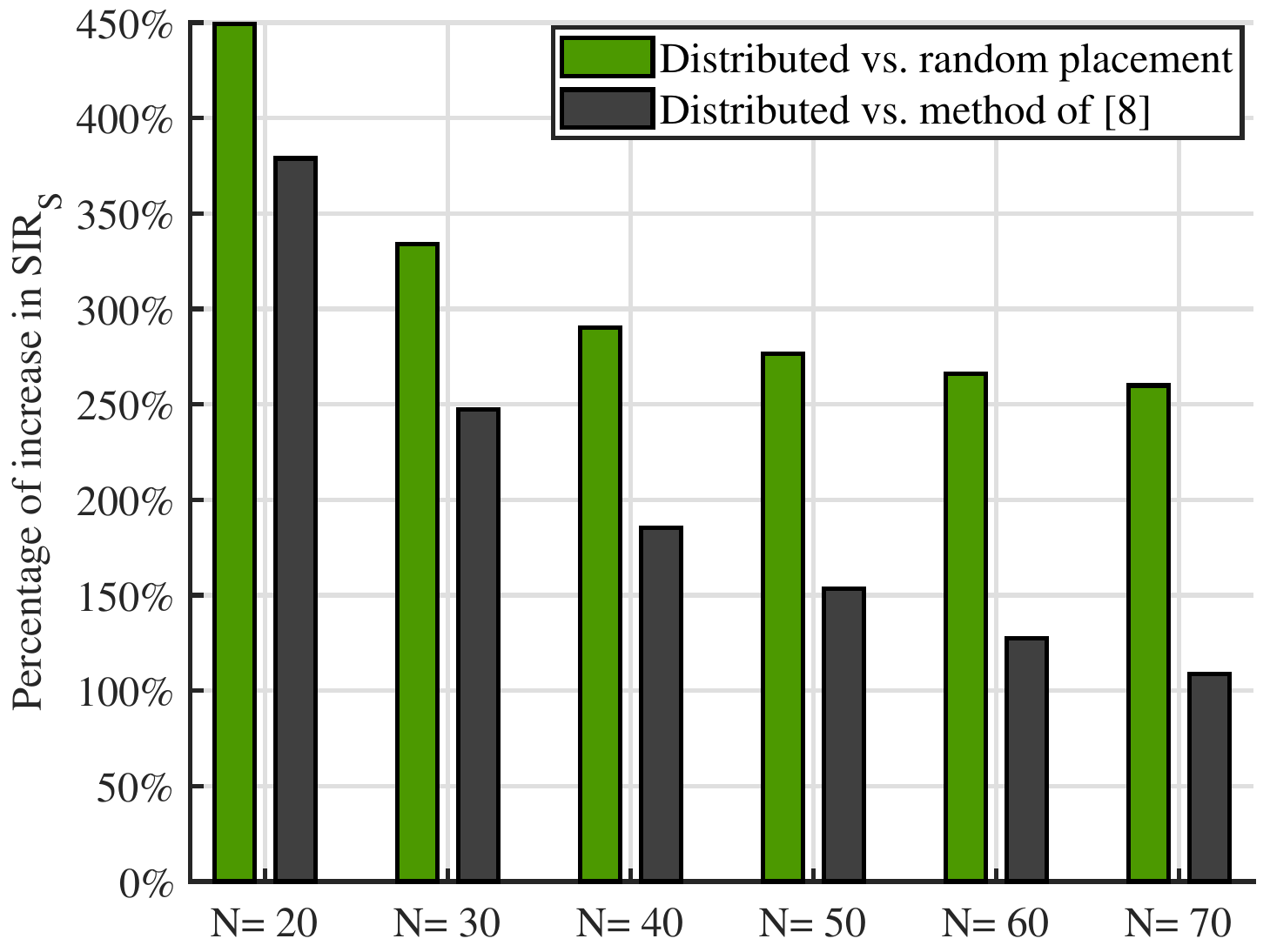}		\caption{Percentage of increase in $\textrm{SIR}_S$ by comparing our distributed algorithm with the random placement and the method in~\cite{chen2018multiple} for various number of UAVs.\label{fig:gainDist}}
		\endminipage
		\vspace{-3.0mm}
	\end{figure}

  Considering $p_t=p_{_{\textrm{MSI}}}=80\textrm{W},D=1000\textrm{m},h=20\textrm{m}$, based on~\eqref{eq:Opt}, the minimum required number of UAVs to satisfy different values of $\textrm{SIR}_S$ when $p_u=1\textrm{W}$ for multiple MSI positions  (when $X_{_{\textrm{MSI}}}=500\textrm{m},Y_{_{\textrm{MSI}}}=400\textrm{m}$ for multiple $p_u$ values) is depicted in Fig.~\ref{fig:minpos} (Fig.~\ref{fig:minpower}). From Fig.~\ref{fig:minpos}, it can be observed that as the MSI gets closer to the Tx/Rx the required number of UAVs increases. Also, from Fig.~\ref{fig:minpower}, it can be seen that by increasing the $p_u$ the required number of UAVs decreases. 
   
  Considering $p_u=1\textrm{W}, p_t=p_{_{\textrm{MSI}}}=80\textrm{W},X_{_{\textrm{MSI}}}=500\textrm{m},Y_{_{\textrm{MSI}}}=400\textrm{m},h=20\textrm{m},D=1000\textrm{m}, d_{min}=4\textrm{m}$, and $\epsilon=3\textrm{m}$ (see Algorithm~\ref{alg:fulldist}), Fig.~\ref{fig:conv} depicts the performance of our distributed algorithm for various number of UAVs in the network. Initially, the UAVs are partitioned into two sets with equal sizes, from which one is placed above the Tx and the other above the Rx. From Fig.~\ref{fig:conv}, it can be seen that as the number of UAVs increases our algorithm achieves larger values of $\textrm{SIR}_S$ with a faster convergence speed. The faster convergence is due to the larger coverage length when having a large number of UAVs. Fig.~\ref{fig:gainDist} reveals the significant increase in the $\textrm{SIR}_S$ obtained through comparing the achieved $\textrm{SIR}_S$ using our distributed algorithm with both the method described in~\cite{chen2018multiple} and the random placement, the performance of which is obtained by randomly placing the UAVs in $1000$ Monte-Carlo iterations.
  \vspace{-2.5mm}
\section{Conclusion}
\noindent
In this work, we studied the UAV-assisted wireless communication paradigm considering the presence of a major source of interference. We proposed a theoretical approach aiming to maximize the SIR of the system considering a single UAV in the network. We extended our study to the Multi-hop scheme, in which utilization of multiple UAVs is feasible. To this end, we first proposed a theoretical approach which simultaneously determines the minimum number of needed UAVs and their optimal positions so as to satisfy a desired SIR of the system. Second,  for a given number of UAVs in the network, we proposed a distributed algorithm requiring message exchange only between the adjacent UAVs so as to maximize the SIR of the system. Furthermore, we illustrated the performance of our methods through numerical simulations.
  \vspace{-2.5mm}
\bibliographystyle{IEEEtran}
\bibliography{ABSbib}
\end{document}